\documentclass{article}
\usepackage{spconf,amsmath,epsfig}
\usepackage{latexsym,amssymb}
\usepackage{graphicx,color}
\newtheorem{theorem}{Theorem}
\newtheorem{lemma}{Lemma}
\newtheorem{corollary}{Corollary}
\newtheorem{proposition}{Proposition}
\newenvironment{proof}{\noindent {\bf Proof.}}{$\blacksquare$}
\newcommand{\ket}[1]{\vert #1\rangle}
\newcommand{\bra}[1]{\langle #1 \vert}
\newcommand{\ketbra}[3]{\langle #1 \vert #2 \vert #3 \rangle}
\newcommand{\inner}[2]{\langle #1 \vert #2 \rangle}

\def\1{\mathbf{1}}

\makeatletter
\def\@ssect#1#2#3#4#5{%
  \begingroup \bf\centering{\interlinepenalty \@M \uppercase{#5}\par}\endgroup%
  \@tempskipa #3\relax \@xsect{\@tempskipa}}
\def\thebibliography#1{\section*{References}\list
 {[\arabic{enumi}]}{\settowidth\labelwidth{[#1]}\leftmargin\labelwidth
 \advance\leftmargin\labelsep
 \usecounter{enumi}}
 \def\newblock{\hskip .11em plus .33em minus .07em}
 \sloppy\clubpenalty4000\widowpenalty4000
 \sfcode`\.=1000\relax}
\makeatother

\title{CURVATURE OF QUANTUM RINGS}

\name{Edmond Jonckheere$^1$\thanks{E.\,Jonckheere was supported by the One Wales Research Institute for Visual Computing (RIVIC), the US National Science Foundation under Grant NetSE 1017881, and ARO MURI grant W911NF-11-1-0268.},
Frank C.\ Langbein$^2$\thanks{F.C.\,Langbein was supported by One Wales Research Institute for Visual Computing (RIVIC).}
 and
Sophie G.\ Schirmer$^3$\thanks{S.G.\,Schirmer was supported by EPSRC ARF Grant EP/D07192X/1 and Hitachi.}}
\address{$^1$USC Center for Quantum Information Science \& Technology, Los Angeles, CA 90089, USA\\
         $^2$School of Computer Science and Informatics, Cardiff University, 5 The Parade, CF24 3AA, UK\\
         $^3$College of Science (Physics), Swansea University, Singleton Park, Swansea, SA6 7JY, UK}

\begin{document}
%
\maketitle
\begin{abstract}
We develop a geometric approach to spin networks with Heisenberg or XX
coupling. Geometry is acquired by defining a distance on the discrete
set of spins. The key feature of the geometry of such networks is their
Gauss curvature $\kappa$, viewed here as the ability to isometrically
embed the chain in the standard Riemannian manifold of curvature
$\kappa$. Here we focus on spin rings. Even though their visual geometry
is trivial, it turns out that the geometry they acquire from the quantum
mechanical distance is far from trivial.
\end{abstract}

\section{Introduction}
\label{sec:intro}

We consider 1-dimensional arrays of $N$ spins arranged in a ring
structure with either Heisenberg or XX interaction specified by the
Hamiltonian
\begin{align*}
H =& \sum_{i=1}^{N-1} J_{i,i+1}
      \left(\sigma^x_i\sigma^x_{i+1} + \sigma^y_i\sigma^y_{i+1}
                             + \epsilon\sigma^z_i\sigma^z_{i+1} \right)\\
   & +J_{N,1} \left(\sigma^x_N\sigma^x_{1}+\sigma^y_N\sigma^y_{1}
                                   +\epsilon\sigma^z_N\sigma^z_{1}\right)
\end{align*}
where $\epsilon=0$ for XX coupling and $\epsilon=1$ for Heisenberg
coupling, which we shall denote by $H_{\rm H}$ and $H_{\rm XX}$ in the
following. The term $J_{N,1}(\ldots)$ represents the coupling energy
between the two ends, spin \#$1$ and \#$N$, of the linear array,
closing the ring. The factor $\sigma^{x,y,z}_i$ is the Pauli matrix
along the $x,y,$ or $z$ direction of spin \#$i$ in the array, i.e.,
\begin{equation*}
\sigma^{x,y,z}_i =
 I_{2\times 2} \otimes \ldots \otimes I_{2 \times 2} \otimes
 \sigma^{x,y,z}\otimes I_{2\times 2}\otimes \ldots \otimes I_{2 \times 2}
\end{equation*}
where the factor $\sigma^{x,y,z}$ occupies the $i$th position among the
$N$ factors and $\sigma^{x,y,z}$ is either of the single spin Pauli
operators
\begin{equation*}
\sigma^x= \begin{pmatrix} 0 & 1 \\ 1 & 0 \end{pmatrix}, \quad
\sigma^y= \begin{pmatrix} 0 & -\imath \\ \imath & 0  \end{pmatrix}, \quad
\sigma^z= \begin{pmatrix} 1 & 0  \\ 0 & -1 \end{pmatrix};
\end{equation*}
$J_{i,i+1}$ denotes the strength of the coupling between spin \#$i$ and
spin \#$(i+1)$ and is inversely proportional to the cubic power of the
{\it physical} distance between spin \#$i$ and spin \#$(i+1)$, here
taken to be uniform (homogeneous arrays). The main point of this paper
is that the simplicity of the geometry given by the physical distance
between spins hides a much more complicated geometry that the network
acquires via a quantum mechanically relevant distance. The latter is
related to be the (maximum) probability of transmission of an excitation
from one spin to another. Here we restrict our attention to the single
excitation subspace, i.e., it is assumed that the total number of
excitations in the network is one. An excitation is transmitted from one
spin and read out from any other spin.

Through the quantum mechanical distance, the spin network acquires a
geometry completely different from the simple geometry of the physical
arrangement of the spins.  For example, a ring made up of an arbitrary large but even number
$N$ of spins becomes a regular $[N/2]$-simplex for the quantum
mechanical distance.  Beyond this simple illustrative example, here by
{\it geometry} we mean {\it curvature}, which can be defined for either
Riemannian or non-Riemannian spaces~\cite{Gromov2001}.  
For metric possibly non-Riemannian spaces, curvature can be defined via the Gromov
$\delta$ or the scaled Gromov $\delta$
(see~\cite{Gromov1987,BridsonHaefliger1999} for various definitions of
$\delta$ and~\cite{scaled_gromov,upper_bound,4_point} for various
definitions of the scaled $\delta$.)  Recall that the Gromov $\delta$
measures the ``fatness'' of the geodesic triangles, with the idea that
``thin'' triangles are symptomatic of negative curvature whereas ``fat''
triangles are symptomatic of positive curvature.  The Gromov 
$\delta$ approach to curvature of spins in various geometrical
arrangements was done in~\cite{first_with_sophie}.  Here we basically
perform the same analysis but remain closer to the traditional
Riemannian approach and focus on ring structures instead of linear
chains.  More specifically, we investigate whether there exists an
isometric embedding $(\mathcal{V},d) \hookrightarrow
\mathbb{M}^r_\kappa$, where $(\mathcal{V},d)$ is the metric space of the
spins endowed with their quantum mechanical distance and
$\mathbb{M}^r_\kappa$ is the standard $r$-dimensional Riemannian space
of uniform curvature $\kappa$.

\section{THE QUEST FOR A DISTANCE}
\label{sec;distance}

Let $\ket{i}$ be the quantum state where the excitation is on spin
\#$i$. The quantum mechanical probability of transition from state
$\ket{i}$ at time $0$ to state $\ket{j}$ at time $t$ is given by
\begin{equation*}
  p\left(\ket{i,0},\ket{j,t}\right)=|\ketbra{i}{e^{-\imath Ht}}{j}|^2
\end{equation*}
Recall that this formula is a corollary of the Feynman path
integral~\cite{kaku,MIT}.  

In order to derive a metric on the vertex set \linebreak
$\mathcal{V}=\left\{\ket{i}:i=1,\dots,N\right\}$ from the probability data, we
inspire ourselves from a closely related situation in sensor networks.
In sensor networks, $\mathcal{V}$ is the set of sensors and a {\it packet
reception rate} $\mathrm{PRR}(i,j)$ is defined as the probability of
successful transmission of the packets from sensor \#$i$ to sensor \#$j$.
Then as shown in~\cite{eurasip_clustering} a useful ``distance'' is
given by $d(i,j)=-\log \mathrm{PRR}(i,j)$. Should there be a violation
of the triangle inequality, say, $d(i,j) > d(i,k)+d(k,j)$, then the
distance between $i$ and $j$ is redefined as $d(i,k)+d(k,j)$.

We follow the same path here, with the warning that packet transmission
from \#$i$ to \#$j$ follows {\it one} wireless link, whereas quantum
mechanical transition from $\ket{i}$ to $\ket{j}$ follows {\it many}
paths.  Following the approach of~\cite{eurasip_clustering}, we could
define a ``distance'' as $-\log p\left(\ket{i,0},\ket{j,t}\right)$, but
this would make the distance time-dependent.  
To remove the dependency on the time, 
define $\Pi_k$ to be the projector onto the $k$th eigenspace of the Hamiltonian
$H=\sum_k \lambda_k \Pi_k$ and let $k=1,...,N$ correspond to the first excitation subspace $\mathcal{H}_1$. 
Then, as 
in~\cite{first_with_sophie}, we define {\it maximum
transition probability} also referred to as {\it Information Transfer Capacity}: 
\begin{equation}
  p_\mathrm{max}(\ket{i},\ket{j}) :=
  \left|\sum_{k=1}^N \left|\bra{i}{\Pi_k} \ket{j}\right|\right|^2
  \ge 
  p\left(\ket{i,0},\ket{j,t}\right)
\end{equation}
Furthermore, as proved in~\cite{first_with_sophie}, under the condition that
$\lambda_k/\pi: k=1,...,N$  
are rationally independent, the maximum transition probability can be reached:
%
$ \sup_{t \geq 0} p\left(\ket{i,0},\ket{j,t}\right) = p_{\max}(\ket{i},\ket{j})$.
%
Although we cannot in general expect \linebreak 
$-\log p_{\max}(\ket{i},\ket{j})$ to satisfy the usual
requirements for a distance, we show that this is the case for certain types of
networks, which allows us to study their geometry and curvature with regard to 
this metric.

\section{UNIFORM SPIN RINGS}
\label{sec:rings}

Unlike for linear chains, the single excitation subspace Hamiltonians
for rings with uniform XX and Heisenberg coupling differ only by a
multiple of the identity, which does not affect the distance.  Hence,
the analysis is the same for both of these physically relevant cases.
The properties of $-\log p_{\max}(\ket{i},\ket{j})$ and the behavior of the distance with $N$ are given
by the following:
\begin{theorem}
\label{t:distance_properties} For a quantum ring ${\cal R}_N$ of $N$
uniformly distributed spins with XX or Heisenberg couplings,
$d_N(i,j):=-\log p_{\max}(\ket{i},\ket{j})$ has the following properties:
\begin{enumerate}
\item For $N$ odd $({\cal R}_N,d_N)$ is a metric space.  

\item For $N$ even $({\cal R}_N,d_N)$ is a semi-metric space that
      becomes metric after antipodal point identification.

\item If $N=p$ or $N=2p$, where $p$ is a prime number, then the
      distances on the space of equivalence classes of spins are
      uniform, i.e., $d_N(i,j)=c_N$ for $i\neq j$.  Otherwise, the
      distances are non-uniform.

\item In all cases $\lim_{N \to \infty} d_N(i,j)=2\log\tfrac{\pi}{2}$, $i\not= j \bmod (N/2)$. (See Fig.~\ref{fig:var_dist_ring} for an illustration.)
\end{enumerate}
\end{theorem}

\begin{proof}
To show that $({\cal R}_N,d_N)$ is a semi-metric space we need to verify
that (i) $d_N(i,i)=0$, (ii) $d_N(i,j)=d_N(j,i)$ and (iii) the triangle
inequality holds.  For a metric space we must further have (iv)
$d_N(i,j)\neq 0$ unless $i=j$.

(i) is clearly satisfied as the projectors onto the eigenspaces are a
resolution of the identity, $\sum_k\Pi_k = I$, and thus for any unit
vector $\ket{i}$, we have $\sum_{k=1}^N |\bra{i}\Pi_k\ket{i}|=
\sum_{k=1}^N |\Pi_k \ket{i}|^2 =1$.  (ii) follows from
$|\bra{i}\Pi_k\ket{j}|=|\bra{j}\Pi_k\ket{i}|$.  The proof of the
remaining properties relies on the circulant matrix property of the
Hamiltonian $H_1$ in the first excitation subspace $\mathcal{H}_1$.  

Specifically, $(H_1)_{i,i+1}=(H_1)_{i+1,i}=h$ for $i=1,\dots,N-1$,
$(H_1)_{N,1}=(H_1)_{1,N}=h$ and $(H_1)_{i,j}=0$ everywhere else.  The eigenvalues
and eigenvectors of $H_1$ are
\begin{subequations}
 \begin{align}
 \lambda_k & =  h(\rho^k_N+\rho^{k(N-1)}_N)
             = 2h\cos\left(\tfrac{2\pi{k}}{N}\right) \\
 w_k       &= \tfrac{1}{\sqrt{N}} \left(1,\rho^k_N, \rho^{k2}_N, 
                                  \ldots, \rho^{k(N-1)}_N\right)^T
\end{align}
\end{subequations}
for $k=0,\dots,N-1$, where $\rho^k_N:=e^{\frac{2\pi \imath k}{N}}$ are
$N$th roots of unity.  Observe the double eigenvalues $\lambda_k=\lambda_{N-k}$ 
except for $k=0$, and $N/2$ if $N$ even.  Thus, each of these double eigenvalues has two complex conjugate
eigenvectors $v_k$ and $v_k^*$.  These eigenvectors need not be orthogonal
but observing that $\inner{w_k}{w_\ell}=\delta_{k\ell}$ and $\inner
{w_k}{w_k^*}=0$, shows that
\begin{equation}
\begin{split}
v_0 &= w_0 = \tfrac{1}{\sqrt{N}}(1,1,\ldots)^T\\ 
v_k &= w_k, \quad v_{N-k} =w_k^*, \quad k=1,\ldots N'=\lfloor\tfrac{N-1}{2} \rfloor \\ 
v_{N/2} &= w_{N/2} = \tfrac{1}{\sqrt{N}}(1,-1,\ldots )^T, \quad \mbox{if N is even}
\end{split}
\end{equation}
defines an orthonormal basis of $\mathcal{H}_1$.  Furthermore, in the
basis in which $H_1$ is circulant, $\ket{i}=e_i$, where $\{e_i:i=1,...,N\}$ 
is the natural basis of $\mathbb{C}^N$. 
We have
\begin{align*}
 |\bra{i}\Pi_0\ket{j}| 
 &= |\inner{i}{v_0}\inner{v_0}{j}| = \tfrac{1}{N}\\
 |\bra{i}\Pi_k\ket{j}|
 &= |\inner{i}{v_k}\inner{v_k}{j} +
     \inner{i}{v_{N-k}}\inner{v_{N-k}}{j}| \\
 &= |\rho_N^{ki} (\rho_N^{kj})^* +
     (\rho_N^{ki})^* \rho_N^{kj}|\tfrac{1}{N} \\
 &= |\rho_N^{k(i-j)}+\rho_N^{-k(i-j)}|\tfrac{1}{N} = \tfrac{2}{N}\left|\cos(\tfrac{2\pi k(i-j)}N)\right|\\
|\bra{i}\Pi_{N/2}\ket{j}|
 &= | \inner{i}{v_{N/2}}\inner{v_{N/2}}{j}| = \tfrac{1}{N}.
\end{align*}
Summing over all eigenspaces $k=0,\ldots,\lfloor N/2 \rfloor$ gives
\begin{align}
\label{eq:pmax}
 & \sqrt{p_{\max}(\ket{i},\ket{j})} \nonumber\\
=& \left\{ \begin{array}{ll}
    \frac{1}{N}+\frac{2}{N}\sum_{k=1}^{N'}\left|\cos\left(\tfrac{2\pi k(i-j)}{N}\right)\right|,
    & N=2N'+1 \\
    \frac{2}{N}+\frac{2}{N}\sum_{k=1}^{N'}\left|\cos\left(\tfrac{2\pi k(i-j)}{N}\right)\right|,
    & N=2N'+2 
\end{array} \right.
\end{align}
For $i=j$ all cosines in Eq.~(\ref{eq:pmax}) are equal to 1 and we have
$\sqrt{p_{\max}(\ket{i},\ket{i})}=(1+2N')/N=1$ for $N=2N'+1$ and
$\sqrt{p_{\max}(\ket{i},\ket{i})}=(2+2N')/N=1$ for $N=2N'+2$, which shows that
$d(i,i)=0$.  For $N=2N'+1$ it is easy to see that $p_{\max}(i,j)=1$ if
and only if $i=j$, hence (iv).  For $N=2N'+2$, on the other hand, we also have
$|\cos(\tfrac{2\pi k N/2}{N})|=|\cos(\pi k)|=1$, and thus $d(i,j)=0$ for
$i-j=N/2$, i.e., the distance vanishes for antipodal points, and thus
$d(i,j)$ is at most a semi-metric.  However, noting that
$d(i,j)=d(i,N'+1+j)$ for $j\le N'+1$, we can identify antipodal points
$\ket{j}$ and $\ket{j+N'+1}$ and let $d$ be defined on the set of
equivalence classes $[\ket{j}]$ for $j=1,\ldots, N'+1$.

To show that the triangle inequality is satisfied, we show that
$\sqrt{p_{\max}(\ket{\ell},\ket{m})} \sqrt{p_{\max}(\ket{m},\ket{n})} \le
\sqrt{p_{\max}(\ket{\ell},\ket{n})}$.  
From the definition of $p_{\mathrm{max}}$ in terms of the eigenvectors of $H_1$ 
we have
\begin{align*}
  \sqrt{p_{\max}(\ket{\ell},\ket{m})} &= \frac{1}{N} \sum_{k=0}^{N-1}  \alpha_k  \rho_N^{k(m-\ell)} \\
  \sqrt{p_{\max}(\ket{m},\ket{n})} &= \frac{1}{N} \sum_{k'=0}^{N-1} \beta_{k'}\rho_N^{k'(n-m)} 
\end{align*}
where $\alpha_k,\beta_{k'}=\pm1$.  Setting
\begin{equation*}
 \gamma_k = \sum_{k'=0}^N \alpha_k \beta_{k'}\rho_N^{(k'-k)(n-m)}
\end{equation*}
we obtain
\begin{align*}
   & \sqrt{p_{\max}(\ket{\ell},\ket{m})} \sqrt{p_{\max}(\ket{m},\ket{n})} \\
 = & \frac{1}{N^2} \sum_{k,k'=0}^{N-1} \alpha_k \beta_{k'}
     \rho^{k(m-\ell)}_N \rho^{k'(n-m)}_N \\
 = & \frac{1}{N^2} \sum_{k,k'=0}^{N-1} \alpha_k \beta_{k'}
     \rho_N^{k(n-\ell) + (k'-k)(n-m)} \\ 
 = & \frac{1}{N^2} \sum_{k=0}^{N-1} \gamma_k \rho_N^{k(n-\ell)} 
 = \left|\frac{1}{N^2} \sum_{k=0}^{N-1} \gamma_k \rho_N^{k(n-\ell)} \right|.
\end{align*}
The final equality follows because the LHS and thus the RHS are known to
be real and positive.  Furthermore, as $\rho_N$ is a root of unity,
$|\rho_N|=1$, and recalling $|\alpha_k|=|\beta_{k'}|=1$,
\begin{align*}
  |\gamma_k| 
    &= \left|\rho_N^{k(m-n)} \sum_{k'=0}^{N-1} \alpha_k \beta_{k'}\rho_N^{k'(n-m)} \right| \\
    &\le \left|\rho_N^{k(m-n)}\right|\cdot \sum_{k'=0}^{N-1} \left|\alpha_k\beta_{k'}\rho_N^{k'(n-m)} \right| = N.
\end{align*}
Again we have $\rho_N^{(N-k)(m-n)}=\rho_N^{-k(m-n)}$, and as the
LHS above is known to be real, we know that we must have
$\gamma_{k}=\gamma_{N-k}$.  Hence, we can again collect exponential
terms pairwise to obtain cosines, which gives for $N=2N'+1$:
\begin{align*}
 \left|\frac{1}{N^2} \sum_{k=0}^{N-1}\gamma_k \rho_N^{k(n-\ell)}\right|
 &= \left|\frac{\gamma_0 }{N^2} 
    + \frac{1}{N^2}\sum_{k=1}^{N'} 2 \gamma_k \cos\left(\tfrac{2\pi k(n-\ell)}{N}\right) \right|\\
 &\le \frac{|\gamma_0|}{N^2} 
    + \frac{2}{N^2}\sum_{k=1}^{N'}|\gamma_k| \left|\cos\left(\tfrac{2\pi k(n-\ell)}{N}\right)\right| \\
 &\le \frac{1}{N} +
 \frac{2}{N}\sum_{k=1}^{N'}\left|\cos\left(\tfrac{2\pi k(n-\ell)}{N}\right)\right| \\
 &= \sqrt{p_{\max}(\ket{\ell},\ket{n})}
\end{align*}
For $N=2N'+2$ we simply replace $\gamma_0$ by $\gamma_0+\gamma_{N'+1}$
above to obtain
\begin{align*}
 \left|\frac{1}{N^2} \sum_{k=0}^{N-1}\gamma_k \rho_N^{k(n-\ell)}\right|
 &\le \frac{2}{N} +
 \frac{2}{N}\sum_{k=1}^{N'}\left|\cos\left(\tfrac{2\pi k(n-\ell)}{N}\right)\right| \\
 & = \sqrt{p_{\max}(\ket{\ell},\ket{n})}.
\end{align*}
This proves (iii) and hence parts (1) and (2) of the theorem. 

To establish (3) we note that if $N=2N'+1$ is prime then 
\begin{equation*}
  \sum_{k=1}^{N'}\left|\cos\left(\tfrac{2\pi k(i-j)}{N}\right)\right| 
 = \sum_{k=1}^{N'}\left|\cos\left(\tfrac{2\pi k}{N}\right)\right| 
\end{equation*}
If $N$ is not $p$ or $2p$ then $N$ and $(i-j)$ will have factors (which
can be canceled) in common for some $(i-j)$ but not for others and
hence we will obtain different distances.  

To establish (iv) letting $N \rightarrow \infty$, it is easily seen that the dependency on
$i,j$ is eliminated provided $i \not= j \bmod (N/2)$. Hence, taking the
norm of the above and then $-\log(\cdot)$ it follows that, for infinite
rings, the distance is uniform for $i \not= j + \bmod (N/2)$.  Finally,
\begin{eqnarray*}
\lefteqn{\lim_{N\to \infty}\sqrt{p_{\mathrm{max}}(\ket{i},\ket{j})}=\lim_{N\to\infty} \frac{2}{N} \sum_{k=0}^{N/2} |\cos((i-j)2\pi k/N)|}~~~~~~~~~~~~~~~~~~~~~~~~~~~~~~\\
&=& \frac{4|i-j|}{2\pi}\int_{0}^{\frac{\pi}{2|i-j|}} \cos(|i-j| x) dx \\
&=& \frac{2|i-j|}{\pi |i-j|} [\sin(|i-j| x)]_0^{\frac{\pi}{2|i-j|}} = \frac{2}{\pi}~~
\end{eqnarray*}
shows that the limiting value for the distance is $d_\infty(i,j)=-2 \log \tfrac{\pi}{2}
\approx 2 \times 0.4516$ for $i \not= j \bmod(N/2)$. 
\end{proof}

The preceding theorem points to discrete spaces which, possibly after
some identification, become complete order $n$ graphs $K_n$ with uniform or nearly
uniform link weight. Complete graphs with uniform link weight are among
the very few that are embeddable in {\it all} spaces of {\it constant}
curvature: negatively curved Riemannian manifolds
$\mathbb{H}_{\kappa<0}$, Euclidean spaces $\mathbb{E}$, and positively
curved Riemannian manifolds $\mathbb{S}_{\kappa>0}$.  The latter
embedding appears the most natural, since the various vertices are
nearly uniformly ``filling'' the whole sphere, whereas it is impossible
to uniformly fill a Euclidean or a hyperbolic space with finitely many
vertices, as $\mathbb{S}_\kappa$ and $\mathbb{H}_\kappa$ are infinite.
Also $K_n$ can be embedded in a sphere of dimension $(n-2)$, whereas the
smallest dimension in which a nonpositively curved manifold can contain
$n$ equidistant points is $(n-1)$.  This not to say that it is futile to
consider embeddings in, say, hyperbolic spaces; indeed, hyperbolic
spaces have specific transport phenomena that could map to the graphs
they support but we begin with embeddability in uniformly positively
curved spaces.

The following theorem makes this embeddability precise. For notational
convenience, let $\widetilde{\cal R}_N$ denote the ring ${\cal R}_N$
after anti-podal identification. Also define the following function
\begin{equation*}
  \kappa_{\mathrm{max}}(n,w):= \left[w^{-1} \cos^{-1}\left(-\tfrac{1}{n-1}\right)\right]^2
\end{equation*}
where $w$ is typically some edge weight.

\begin{theorem}
The spin rings are embeddable in the following spaces:
\begin{enumerate}
\item If $N$ is an odd composite number, $({\cal R}_N,d_N)$ is isometric
      to a complete graph $K_{N}$ with nonuniform link weight;
      furthermore, for $N$ large enough, there exists an isometric
      embedding $({\cal R}_N,d_N) \hookrightarrow
      (\mathbb{S}^{N-1}_\kappa,d_\kappa)$ for $\kappa \leq
      \kappa_{\mathrm{max}}(N,\bar{w})$.

\item If $N=2c$ where $c$ is a composite number, $(\widetilde{\cal
      R}_N,d_N)$ is isometric to a complete graph $K_{N/2}$ with
      nonuniform link weight; furthermore, for $N$ large enough, there
      exists an isometric embedding $({\cal R}_N,d_N) \hookrightarrow
      (\mathbb{S}^{(N/2)-1}_\kappa,d_\kappa)$ for $\kappa \leq
      \kappa_{\mathrm{max}}(N/2,\bar{w})$.

\item \textbf{\textit{a.}} If $N=p$, where $p$ is a prime number, $({\cal R}_N,d_N)$ is isometric to a complete
      graph $K_N$ of uniform edge weight $w=d_N(i,j)$, $i \not= j$;
      furthermore, there exists an isometric embedding $({\cal R}_N,d_N)
      \hookrightarrow (\mathbb{S}^{N-1}_\kappa,d_\kappa)$ in the
      $(N-1)$-dimensional sphere of curvature $\kappa\leq
      \kappa_{\mathrm{max}}(N,w)$.  Finally, there exists an irreducible
      isometric embedding $({\cal R}_N,d_N) \hookrightarrow
      (\mathbb{S}^{N-2}_\kappa,d_\kappa)$ in the sphere of curvature
      $\kappa=\kappa_{\mathrm{max}}(N,w)$.\\
      \textbf{\textit{b.}} If $N=2p$ where $p$ prime, $(\widetilde{\cal R}_N,d_N)$ is
      isometric to a complete graph $K_{N/2}$ with uniform link weight;
      furthermore, there exists an embedding $(\widetilde{\cal R}_N,d_N)
      \hookrightarrow (\mathbb{S}^{N/2-1},d_\kappa)$ for $\kappa\leq
      \kappa_{\mathrm{max}}(N/2,w)$ and an irreducible isometric
      embedding $(\widetilde{\cal R}_N,d_N) \hookrightarrow
      (\mathbb{S}^{N/2-2},d_\kappa)$ for $\kappa =
      \kappa_{\mathrm{max}}(N/2,d)$.

\end{enumerate}
\end{theorem}

\begin{proof}
Part 3.a of the proof is in the Appendix.

Part 3.b of the proof is a corollary of the Appendix.  By doing the
anti-podal identification on the semi-metric space $({\cal R}_N,d_N)$,
one obtains the metric space $(\widetilde{\cal R}_N,d_N)$.  The latter
is clearly isomorphic to $(K_{N/2},w)$, where $w$ is the uniform link
weight.  The result then follows by applying the Appendix to
$(K_{N/2},w)$.

Part 1 relies on the continuity of the Gram matrix relative to distance
data. Define $(K_\infty,w_\infty)$ be the complete graph on countably
infinitely many vertices with uniform link weight $w_\infty:=\lim_{N\to
\infty}d_N(i,j)$, $i \not= j$. By a limiting argument on the Appendix,
it follows that $(K_\infty,w_\infty)$ is isometrically embeddable in an
infinite-dimensional sphere (itself embedded in the Hilbert space
$\ell^2$ as $\sum_{i=1}^\infty x_i^2=1/\sqrt{\kappa}$) of curvature
$\kappa\leq \kappa_{\mathrm{max}}(\infty,w_\infty)=\pi^2/4w_\infty^2$.
Hence the associated Gram matrix $G_\kappa(D_\infty(1,2,3,\dots))$ for
uniform link weight is positive definite.  Trivially, the $N \times N$
section of the Gram matrix $G_\kappa(D_\infty(1,2,\dots))_{N \times N}$
for uniform link weight is also positive definite.  Since $\lim_{N \to
\infty} G_\kappa(D_N(1,\dots,N))=G_\kappa(D_\infty(1,2,\dots))_{N \times
N}$, and since the latter is positive definite, there exists a $N$ large
enough such that $G_\kappa(D_N(1,\dots,N))$ $>0$. The latter means that
the ring ${\cal R}_N$ is isometrically embeddable in
$\mathbb{S}^{N-1}_\kappa$.

Part 2 follows from a combination of the argument of Part 3.b and Part 1.
\end{proof}

\begin{figure}[t]
 \centering
\centerline{\includegraphics[width=.8\columnwidth]{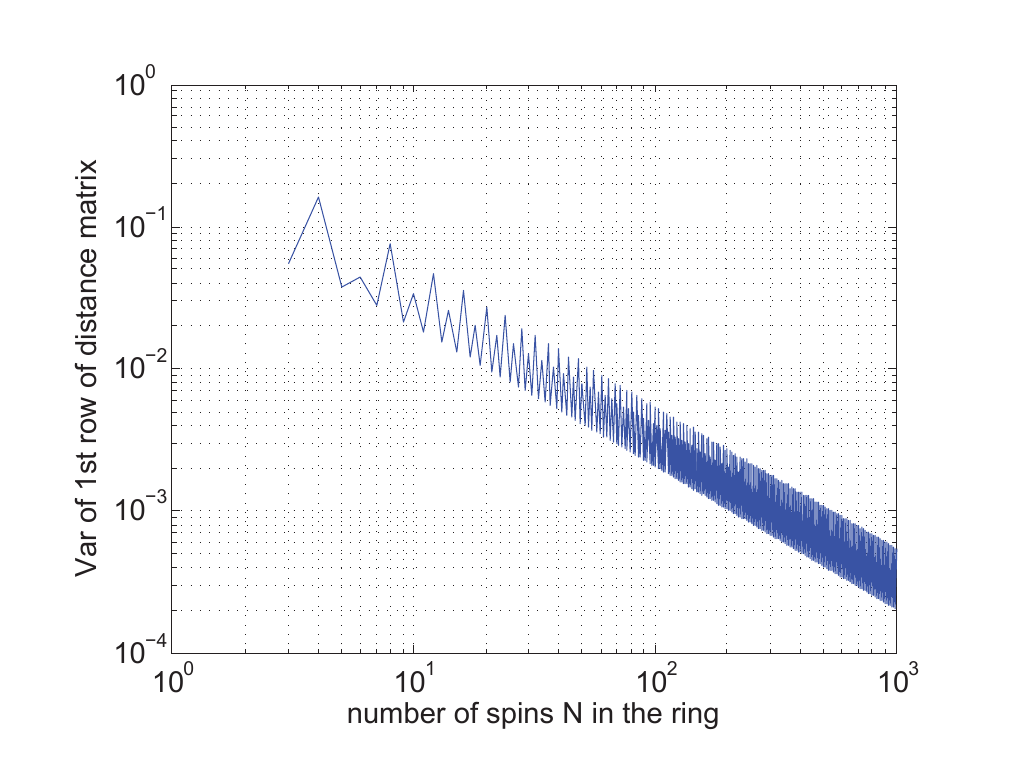}}
 \caption{Variance of quantum mechanical distance between spins showing its
   decrease as the number of spins increases.}
\label{fig:var_dist_ring}
\end{figure}

\section{Conclusion}

In conclusion, besides their 1-dimensional physical geometry, {\it
quantum} rings have been shown to have higher dimensional geometry for
the quantum mechanical distance.  In forthcoming work will investigate
how the geometry can be changed to improve transmission fidelities by
local control.

\section*{Appendix: Embeddability of complete graph in constant curvature spaces}

Given a set $\mathcal{V}$ of $n$ vertices  and the distance data
$D(1,2,\dots,\linebreak n) =\left\{ d(i,j) \right\}_{i,j=1,\dots,n}$, isometric
embedding of $(\mathcal{V},D)$ in a Riemannian manifold of uniform curvature
$\kappa \not= 0$ involves the {\it Gram matrix} $G_\kappa(D(1,\ldots,n))$
(we sometimes simplify this to $G_\kappa(D)$).  In the positive curvature case,
$G_{\kappa>0}(D) = \left\{\cos\sqrt{\kappa}d(i,j)\right\}_{i,j=1,\dots,n}$
and in the negative curvature case
$G_{\kappa<0}(D)=\left\{\cosh\sqrt{-\kappa}d(i,j)\right\}_{i,j=1,\dots,n}$.
$(\mathcal{V},D)$ is embeddable in $\mathbb{S}^{N-1}_\kappa$ iff $\kappa \leq
\frac{\pi^2}{\max_{i,j} d^2(i,j)}$ and $G_\kappa(D)>0$
(see~\cite{Blumenthal1953}).  $G_\kappa(D)>0$ is of course equivalent to
$\det G_\kappa(D(1,\dots,k))>0$, $\forall k=1,\dots,n$.  $(\mathcal{V},D)$ is
embeddable in $\mathbb{H}_{\kappa<0}^{n-1}$ iff $\mathrm{sign} \det
G_\kappa(D(1,\dots,k))=(-1)^{k+1}$ (see~\cite{Blumenthal1953}).

Embeddability in Euclidean space involves the Cayley-Menger matrix
$CM(D(1,\dots,n))=$\\
\centerline{$\begin{pmatrix}
0 & 1            & 1            & \ldots  &  1 \\
1 & 0            & d(1,2)^2 & \ldots  & d(1,n)^2 \\
1 & d(2,1)^2 & 0            & \ldots  & d(2,n)^2 \\
\vdots & \vdots & \vdots & \ddots & \vdots \\
1 & d(n,1)^2 & d(n,2)^2 & \ldots  & 0
\end{pmatrix}$.}
Embeddability in Euclidean space is equivalent to \linebreak
$\mathrm{sign}\det CM(D(1,\dots,k))=(-1)^k$, $k=2,\dots,n$
(see~\cite{Blumenthal1953}).

Embeddability of the complete graph with uniform link weight in both the
constant curvature hyperbolic space and the constant curvature spherical
space involves the special $k \times k$ Toeplitz structure\\
\centerline{$T_k = \begin{pmatrix}
1 & c & \ldots & c \\
c & 1 & \ldots & c \\
\vdots & \vdots & \ddots & \vdots \\
c & c & \ldots & 1
\end{pmatrix}, \quad k \geq 1$}
where $c=\cosh \left( d(i,j)\sqrt{-\kappa}\right)$ in the hyperbolic case and
$c=\cos \left( d(i,j)\sqrt{\kappa}\right)$ in the spherical case. The issue is the
sequence of principal minors of such a Toeplitz-structured matrix. Set
$t_k=\det T_{k \times k}$
and we have the following lemma:
\begin{lemma}
\label{l:recursion1} The recursion on the principal minors of the
Toeplitz-structured matrix $T_k$ is
\begin{equation*}
  t_{k+1}=(1-c)t_k+(1-c)^2t_{k-1}-(1-c)^3t_{k-2}
\end{equation*}
subject to the initial conditions
\begin{align*}
t_1 & = 1 \\
t_2 & = 1-c^2 \\
t_3 & = (1-c)^2(2c+1).
\end{align*}
Furthermore, the solution to the above recursion is given by
\begin{equation*}
  t_k = (1-c)^{k-1} \left( (k-1)c + 1 \right), \quad k \geq 1.
\end{equation*}
\end{lemma}

\begin{proof}
By subtracting the first column from the last column of $T_k$, we get
\begin{equation*}
  \det T_k = (-1)^{k+1} (c-1) \det T + (1-c) \det T_{k-1}
\end{equation*}
where $T$ is the Toeplitz matrix with $1$'s on the superdiagonal and
$c$'s everywhere else.  Again, by subtracting the first row from the
last row of $T$, we get
\begin{equation*}
\det T= (-1)^{k+1} (c-1) \det \begin{pmatrix}
                          c & c 1^T_{k-3} \\
                          c 1_{k-3} & T_{k-3}
                          \end{pmatrix}
\end{equation*}
where $1_k$ is the $k$-dimensional column made up of $1$'s.
Observing that
\begin{equation*}
  \begin{pmatrix}
       c & c 1^T_{k-3} \\
       c 1_{k-3} & T_{k-3}
  \end{pmatrix}
= \begin{pmatrix}
       (c-1) & 0 \\
          0  & 0
  \end{pmatrix} + T_{k-2}
\end{equation*}
and remembering that the determinant of the sum of two matrices equals the sum
of the determinants of all matrices constructed with some columns of the first matrix
and the complementary columns of the second matrix, we get
\begin{align*}
\det \begin{pmatrix}
  c & c 1^T_{k-3} \\
 c 1_{k-3} & T_{k-3}
\end{pmatrix} &= \det T_{k-2} + \det \begin{pmatrix}
                       (c-1) & c1_{k-3}\\ 0   & T_{k-3}
                       \end{pmatrix} \\
 &= \det T_{k-2} + (c-1) \det T_{k-3}.
\end{align*}
Combining all of the above yields the recursion. The initial conditions
on the recursion are trivial to verify. The explicit solution is easily
seen by direct verification to satisfy the recursion and its initial
conditions.
\end{proof}

From the above, it is possible to say something about the eigenvalues of $T_n$.

\begin{corollary}
\label{c:eigenvalues}
\begin{equation*}
   \det \left( sI-T_n \right)=(s-(1-c))^{n-1}(s-((n-1)c+1)).
\end{equation*}
\end{corollary}

\begin{proof}
Recall that the coefficient of $s^{n-k}$ in $\det \left( sI-T_n \right)$
is $(-1)^k$ times the sum of all principal minors of order $k$ of
$T_n$. There are $\binom{n}{k}$ such principal minors, all equal to
$t_k$. Hence,
\begin{eqnarray*}
\lefteqn{\det \left( sI-T_n \right) = \sum_{k=0}^n (-1)^k \binom{n}{k} t_k s^{n-k}}\\
&=& \sum_{k=0}^n (-1)^k \binom{n}{k} (1-c)^{k-1}\left( (k-1)c+1 \right) s^{n-k} \\
&=&\left( \sum_{k=0}^{n-1} (-1)^k \binom{n-1}{k}(1-c)^{k} s^{n-1-k} \right) \\
&&\quad \quad \times \left( s-((n-1)c+1) \right) \\
&=& (s-(1-c))^{n-1}(s-((n-1)c+1))
\end{eqnarray*}
\end{proof}

We now look at the recursion in the case of the Cayley-Menger matrix.

\begin{lemma}
\label{l:cayley_menger}
The recursion on the $k \times k$ top left-hand corner principal minors
$\mathrm{cm}_k:=\det\left(CM_{k\times k}\right)$ of the Cayley-Menger
matrix for uniform distance $d$ is given by
\[ \mathrm{cm}_k=-\left( \frac{(k-1)}{d^2(k-2)}\right) t_{k-1} \]
where the recursion on $t_{k-1}$ is
\[ t_{k-1}=-\left( \frac{(k-2)d^2}{k-3}\right) t_{k-2}. \]
\end{lemma}

\begin{proof}
Applying the Schur lemma to
\[\mathrm{CM}_{k \times k}=\left(\begin{array}{cc}
0 & 1^T_{k-1}\\
1_{k-1} & T_{k-1}
\end{array}\right)\]
where $T_{k-1}$ is the Toeplitz matrix with $0$'s on the diagonal and $d^2$'s
everywhere else, we get
\[ \mathrm{cm}_k=- \left( 1^T_{k-1} T^{-1}_{k-1} 1_{k-1} \right) \det (T_{k-1}).  \]
If we now observe that $1_{k-1}$ is an eigenvector of $T_{k-1}$ with
eigenvalue $d^2(k-2)$, it follows that $1^T_{k-1} T^{-1}_{k-1} 1_{k-1}=(k-1)/(d^2(k-2))$
and setting $t_{k-1}:=\det(T_{k-1})$ the first part of the recursion follows.
Next, if we apply exactly the same Schur lemma argument to
\[T_{k-1}=\left(\begin{array}{cc}
0 & d^2 1^T_{k-2}\\
d^2 1_{k-2} & T_{k-2}
\end{array}\right)\]
the second part of the recursion follows.
\end{proof}

\begin{proposition}
The complete graph $K_{n \geq 2}$ with uniform link weight $d(i,j)>0$, $i \not= j$, is
irreducibly isometrically embeddable in
$\mathbb{E}^{n-1}$ and in $\mathbb{H}^{n-1}_{\kappa<0}$. The same graph is
isometrically embeddable in $\mathbb{S}^{n-1}_{\kappa > 0}$ iff
\begin{equation}
\label{e:isometric_embedding}
\kappa \leq \left[d(i,j)^{-1} \cos^{-1} \left( -\tfrac{1}{n-1} \right) \right]^2.
\end{equation}

Furthermore, it is irreducibly
isometrically embeddable in $\mathbb{S}^{n-2}_{\kappa > 0}$
for
\begin{equation}
\label{e:irreducible}
\kappa= \left[ d(i,j)^{-1} \cos^{-1} \left( -\tfrac{1}{n-1} \right) \right]^2.
\end{equation}
\end{proposition}

\begin{proof}
The proof of Euclidean embedding follows at once from Lemma~\ref{l:cayley_menger},
as the latter indeed reveals that \linebreak
$\mathrm{sign} \det\left(\mathrm{CM}_{k\times k}\right)$ is alternating.

For embeddability in hyperbolic space, set \linebreak $c=\cosh \left(
d(i,j) \sqrt{-\kappa} \right)>1$, and then the lemma yields $\det
G_\kappa(D(1,\dots,k))=(1-c)^{k-1}((k-1)c+1)$.  The principal minors of
$G_\kappa(D(1,\dots,n))$ clearly never vanish and their signs have the
required alternating property, from which irreducible isometric
embedding follows.

For embeddability in spherical space, set \linebreak $c=\cos \left(
d(i,j) \sqrt{\kappa} \right) \leq 1$ and then the lemma yields
\linebreak $G_\kappa (D(1,\dots,k))=(1-c)^{k-1}((k-1)c+1)$. Isometric
embeddability is hence equivalent to the sequence $(k-1)c+1$,
$k=1,\dots,n$ being positive, with possibly a vanishing tail. $(k-1)c+1
\geq 0$ is clearly equivalent to
\begin{equation}
\label{e:positive_forall_k}
\kappa \leq \left[ d(i,j)^{-1} \cos^{-1} \left( -\tfrac{1}{k-1} \right)\right]^2
\end{equation}
and since
\begin{equation}
  \cos^{-1} \left(- \tfrac{1}{n-1} \right) < \cos^{-1} \left(
      -\tfrac{1}{k-1} \right) , \quad k < n
\end{equation}
it follows that $\det G_\kappa(D(1,\dots,k))=(1-c)^{k-1}((k-1)c+1)$ could
only possibly vanish for $k=n$ and is positive for $k<n$. Hence the
graph is isometrically embeddable iff~(\ref{e:positive_forall_k}) is
satisfied $\forall k \leq n$, which is equivalent
to~(\ref{e:isometric_embedding}).  The irreducible isometric embedding
in $\mathbb{S}^{n-2}_\kappa$ requires, in addition, that $(n-1)c+1=0$,
which is equivalent to~(\ref{e:irreducible}).
\end{proof}


\vfill\pagebreak
\end{document}